\documentclass[runningheads,a4paper]{llncs}
\usepackage[misc]{ifsym}
\usepackage{amssymb}
\usepackage{graphicx}
\usepackage{amsmath}
\usepackage{placeins}
\usepackage{multirow}
\usepackage{url}
\urldef{\mailsa}\path|{iarieli,ko,rann}@technion.ac.il|
\newcommand{\keywords}[1]{\par\addvspace\baselineskip
\noindent\keywordname\enspace\ignorespaces#1}

\begin{document}

\mainmatter  

\title{The Crowdfunding Game}

\titlerunning{The Crowdfunding Game}

%
%
\author{Itai Arieli%
\and Moran Koren\textsuperscript{(\Letter)}  \and Rann Smorodinsky\thanks{Research supported by GIF research grant no. I-1419-118.4/2017, ISF grant 2018889, Technion VPR grants, the joint Microsoft-Technion e-Commerce Lab, the Bernard M. Gordon Center for Systems Engineering at the Technion, and the TASP Center at the Technion.}}
\authorrunning{The Crowdfunding Game}

\institute{Faculty of Industrial Engineering,\\ Technion - Israel Institute of Technology. Haifa, Israel.
\mailsa\\
}

%
%

\toctitle{Lecture Notes in Computer Science}
\tocauthor{Authors' Instructions}
\maketitle

\begin{abstract}
The recent success of crowdfunding for supporting new and innovative products has been overwhelming with over 34 Billion Dollars raised in 2015. 
In many crowdfunding platforms, firms set a campaign threshold and contributions are collected only if this threshold is reached.	 During the campaign, consumers are uncertain as to the ex-post value of the product, the business model viability, and the seller's reliability. Consumer who commit to a contribution therefore gambles. This gamble is effected by the campaign's threshold. Contributions to campaigns with higher thresholds are collected only if a greater number of agents find the offering acceptable. Therefore, high threshold serves as a social insurance and thus in high-threshold campaigns, potential contributors feel more at ease with contributing.   We introduce the crowdfunding game and explore the contributor's dilemma in the context of experience goods. We discuss equilibrium existence and related social welfare, information aggregation and revenue implications.

\keywords{Crowdfunding, Voting, Equilbrium, Game Theory}
\end{abstract}

\section{Introduction}

The evolution of the `sharing economy' has  made it possible for the general public to invest in early-stage innovative and economically risky projects and products. The success of this fund raising scheme, known as `crowdfunding', is huge. In 2015 over 34 Billion dollars have been raised while predictions for 2017 are as high 90 Billion dollars.%
\footnote{Figures taken from \url{http://crowdexpert.com/crowdfunding-industry-statistics}.}.

Crowdfunding is often used for products or projects in a nascent stage. In a typical campaign the fund raiser (the firm) presents the concept for the product via some video or images of a product mock-up.%
\footnote{Musicians, for example, present rough cuts of songs they plan to include in some future album for which they raise funds. Fringe theater producers would typically provide the basic story and nothing more.}
Next, the firm posts a price (or menu of prices) for the product that is implicitly assumed lower than the future market price. When buying a product the consumer faces two types of risks. First, whether the firm will have enough funds to produce and deliver the product and second, the quality and value of the product is unknown at the time of the campaign.%
\footnote{Empirical studies show that approximately $8\%$ of successful crowdfunding campaigns result in disappointing products while $20\%$ never deliver (\cite{Mollick2015}).}
Buying a product via a crowdfunding campaign is essentially buying a pig in a poke. To mitigate some of these risks the firm sets a campaign goal which serves as a threshold for implementing the campaign. Contributions are collected only if they exceed the goal, which in turn forms some kind of insurance to the potential buyers. This approach is standard in many on-line crowd-funding platform such as Kickstarter  and Indiegogo among others.

Given the nature of many of the products in crowdfunding platforms it is often the case that their value is not known to the firm itself. Therefore crowdfunding campaign serve a double purpose for firms. Beyond raising the money required to move beyond the drawing board stage, many firms use crowdfunding as a litmus test for appreciating the demand side and the value of their offering. From a firm's perspective profits are expected only after the campaign and only if their product is appealing. 

In recent years the methodology known  of {\em lean start-up} has gained popularity among entrepreneurs. A central building block of this concepts requires entrepreneurs to learn the validity of their product offering and the demand side in the earliest stage possible and crowdfunding is often proposed as a means for doing just that. From the firms point of view the campaign serves as a decision tool. The company pursues the new offering if and only if the campaign surpasses the threshold. In this aspect the interest of the firm and the buyers are perfectly aligned and all would like the move forward with the product whenever its value is high. Our primary subject of investigation is whether the crowdfunding campaign format, and in particular the notion of a campaign threshold, serves the players.

The campaign threshold introduces an interesting strategic interaction. A low threshold implies low traction which may result in the firm moving ahead with the product in-spite of insignificant demand. On the other hand, a high goal implies that the product will only be realized if a large portion of the population finds it desirable. In this case, even consumers with discouraging information about the product will likely buy the product because they `free ride' on the insurance induced by the high threshold which `guarantees' that the firm will collect funds only if the traction is high, which is correlated with the actual value of the product. However, if the consumers with discouraging demand flock to buy the product then they pull the rug from under the premise of the campaign threshold as a form of insurance.

In this paper we focus on the analysis of the insurance component. We propose a model where the value of a product is unknown to all parties. Potential buyers receive signals that are correlated with the true value and make their buying decision accordingly. We introduce a model of a crowdfunding campaign that allows us to study the how well can such campaigns sieve the high value products from the low ones.

The model we propose is stripped away from many aspects of crowdfunding campaigns in order to focus on the interplay between the product quality, the buyers information structure and the success of a campaign. In particular we study a common value setting where the product can have one of two values: $-1$ or $1$. Conditional on the value of the product, potential contributors receive an IID symmetric signal. These potential contributors then either buy the risky product at its ex-ante expected price (which is zero) or not buy and receive for value of the non-risky outcome that is set to zero. In particular, we strip the model away the dynamic nature that is inherent to many crowdfunding campaign platforms.

Due to the simplicity of the model it is applicable to other settings, beyond crowdfunding. Consider, for example, the evolution of multi-national institutions such as the UN's International Court of Justice in Hague, the Kyoto Protocol or the Geneva Conventions. To form these institutions a proposal is made and then states can decide whether to join the institution or not. The institutions eventually form, as with our crowdfunding game, only whenever sufficiently many states decided to participate. In a similar vein, the formation of industry standards can be modeled as a crowdfunding game. In future work we plan to study how robust our observations are to the modeling choice. In particular we plan to go beyond the symmetric information setting and beyond the static setting.

\subsection{Toy model and motivating example}

A population of $N$ potential buyers (or players) and a threshold $B$ induce a game of incomplete information. Hereafter we refer to this induced game as the `crowdfunding game'. In this game, Nature randomly chooses one of two states $\Omega=\{H,L\}$ with equal probabilities. If $\omega=H$, then the product is good and the business model is viable. If $\omega=L$, then the product will be disappointing. The realized state of nature is unknown to the players, and determines the value of the proposed good, $v \in \{-1,1\}$.  Given the state of nature each player $i$ receives a signal $s_i \in\{H,L\}$ where $Prob(s_i=\omega) > 0.5$. Signals are IID across players conditional of the state of nature.  After receiving the signals players choose simultaneously whether to contribute (buy the product) or not. The utility of players that choose not to buy is zero, whereas for those  who choose to buy it is $v$, if there are at least $B$ players who whose to buy and $0$ if the number of buyers is less than $B$. In the latter case all players receive a utility of zero.

As we state above, the crowdfunding mechanism provides the players with social insurance. To clarify this argument and better understand our crowdfunding game consider the following two examples. Assume there are three players $|N|=3$ and that the probability of a player receiving the correct signal is $p=0.75$. In this setting we present two scenarios. The first, a `benchmark' where $B=1$, i.e, a single contribution is sufficient for production. In this case every player knows that if she pledges a contribution the product will be supplied, and a second scenario, in which $B=3$ and the product will be supplied only if all players choose the ``buy" action.

In the `benchmark' case, where $B=1$, the decision of each buyer is binding as a single purchase decision is enough to ensure that the product be supplied.  The expected utility of a player with a good signal, $s_i=H$ will be
$$ Pr(\omega=H|s_i)-Pr(\omega=L|s_i)=p-(1-p)$$
and as $p>\frac{1}{2}$, player $i$ who received a signal $s_i=H$ will surely contribute. Similarly, a player with a bad signal will receive a utility of $(1-p)-p<0$ and therefore will surely choose her dominant strategy of opting-out. However in the second case, where $B=3$, players strategies are less trivial. To see that the vector of pure strategies mentioned above no longer constitutes an equilibrium, assume to the contrary that there exists an equilibrium where player with a good signal surely commits and a player with a bad signal opts out.  Now we claim that the player with a bad signal has a profitable deviation. We denote this player with $l\in N$. Since all other players play according to their signals, the product will be supplied only if the two other agents received a high signal. In this case, $l$'s optimal strategy is to buy as well.  Note that in all other cases, player $l$ receives a utility of $0$ as the product will not be supplied. Therefore if $B=3$ and all other players play the pure strategy mentioned above, player $l$ assumes no risk if she pledges to buy, hence opting-out is now dominated by any other strategy in which player $l$ assigns a positive probability to buying, a contradiction to the assumption that the proposed vector of pure strategy is an equilibrium.  In this case, the only symmetric, non-trivial equilibrium is one where players who receive a good signal surely buy and those who receive a bad signal plays a mixed strategy. When $n=3$, $B=3$ and $p=0.75$, player $l$ will assigning a probability of $\lambda=0.302$ to the ``buy" action. In this equilibrium there is a probability of $0.027$ that the product will be supplied even though all three players received low signals.  We show that the intuition and results of these examples hold in a more general settings and provide key insights into the pros and cons of the popular crowdfunding scheme.

\subsection{Main findings}

Recall that a primary goal of firms engaging in crowdfunding is to gauge the viability of a product. This sanity check is typically used in two ways. It is either used internally to make a Go/No-Go decision or it is used, externally. A successful campaign can serve as a means of attracting further and more substantial investments from potential (institutional) investors. To capture these two, inherently different, goals, we introduce two measures of success for a crowdfunding game:
\begin{itemize}
	\item
	The {\em correctness index} of a game is defined as the probability that the game ends up with a the correct decision. That is, the probability the product is funded when its value is high or the probability that the product is rejected when its value is low. The correctness index measures how well the crowdfunding aggregates the private information from the buyers in order to make sure the firm pursues the product only when it is viable.
	\item
	The {\em market penetration index} is the expected number of buyers provided that the product is supplied, i.e, the threshold is surpassed. This number serves as a proxy for success of the campaign as a means to attract further investments.
\end{itemize}

Our results provide limits on the success, in both aspects, of {\bf large} crowdfunding games\footnote{However in Section \ref{sec:finite_tab} we examine our setting in finite populations as well.}. We state and prove three results:
\begin{itemize}
	\item
	We provide a constrictive proof for the existence of a symmetric, non-trivial equilibrium and we show it is unique. In every such equilibrium players with a high signal surely contribute while those with a low signal either decline or take a mixed strategy whereby they contribute at a positive probability, strictly less than one.
	\item
	In large games, we provide a tight bound on the  correctness index which is strictly less than one. Thus, no matter how the campaign goal is set, full information aggregation cannot be guaranteed. We compare this with the efficiency guarantees of majority voting implied by Condorcet Jury Theorem.
	\item
	I large games, we provide a bound on the penetration index and we show that by setting the champaign goal optimally the resulting market penetration is higher than the benchmark case where the campaign goal is set to a single buyer ($B=1$).
\end{itemize}

\subsection{Related Literature}

The lion's  share of the literature on crowdfunding take an empirical approach according to which parameters of statistical models are calibrated to fit with an empirical data from various
crowdfunding platforms (e.g., \cite{Lei2017},\cite{Ellman2014}, and \cite{Yang2016}).
In these papers crowdfunding platforms are modeled as complex systems and simulations are used  to study the underlying dynamics among the various participants: entrepreneurs, funders, and crowdfunding platforms. Moreover, in these papers there are no strategic considerations or equilibrium analysis, which are at the heart of our model.

Several other empirical papers provide statistical and econometric analysis based on data provided from crowdfunding platforms.  Hemer \cite{Hemer2011} provides a rigorous survey on the principles of the crowdfunding mechanism and how these principals are reflected in the data of ten different crowdfunding platforms from different countries.  Yum et al. \cite{Yum2012}, based on data from a Korean platform, `Popfinding.com',  argues that firms use the crowdfunding platform as a means for information gathering. Mollick analyzed a survey conducted by Kickstarter and found that most campaigns reach extreme results, i.e., very few pledges or a massive over subscription \cite{Mollick2014}. In  \cite{Mollick2015} the author use a survey of over $47,000$ backers to conclude that about $9$ percent of successful kickstarter campaigns never deliver. Our theoretical study complements these findings as we show that the asymptotic correctness of crowdfunding campaigns is bounded away from $1$, thus there is a positive probability that `bad product' will slip through the filter.

The vast majority of the theoretical work on crowd-funding focuses on studying the decision of the entrepreneur (see for example \cite{Chemla2016},\cite{Strausz2017} and \cite{Kumar2017}). In \cite{Strausz2017} Strausz studies the vulnerability of crowd-funding platforms to entrepreneurial moral hazard. The paper studies a model where an entrepreneur who faces demand uncertainty elicit funds via a crowdfunding campaign. If the campaign is successful, it may embezzle a part of the invested sum or produce the product. Consumers posses private information as well as personal valuation for the proposed good and may choose whether or not to commit to purchase the good. The paper then examines if one can design an efficient mechanism in which the entrepreneur behave honestly. The main result shows that this occurs only if the expected returns exceed the investment costs.  Chemka and Tinn \cite{Chemla2016} take similar approach and compare two common crowdfunding mechanisms ``All-or-Nothing (AoN)" and ``Keep-it-All (KiA)". In AoN, as in our model, funds are collected only if it reached the pre-determined threshold. While in the KiA, the firm collects all accumulated funds. The paper shows that the AoN scheme dominates over the KiA scheme in terms of efficiency and is less vulnerable to moral hazard. In \cite{Kumar2017}, the authors present a model in which the firm can finance a product by traditional channels or by crowdfunding. If crowdfunding is chosen, then the firm can utilize price discrimination to extract surplus from high-demand consumers. They find that if external finance is more expensive, then the level of price discriminations decreases and thus reducing the cost of capital may cause inefficient allocation as the firm increase prices.

We take a complementary approach and focus our analysis on the demand side of the market. In our model, the utility can be either positive or negative (for example in  case of a dishonest seller). Consumers receive state-dependent signal over the aforementioned state realization and choose whether to commit a contribution. Our work is closest to Alaei, Malekian and Mostagir \cite{Alaei2016}. In their model buyers have a private value and valuations across players are IID. The focus of the paper is on the dynamic aspect of a crowdfunding campaign as they consider a setting in which players take actions sequentially and irreversibly. The main result provides an elegant explanation for the empirical finding whereby crowdfunding campaigns are either extreme successes (with over subscription) or fail miserably (see \cite{Mollick2014}). 

Another related line of research is the Condorcet model, where a set of players with private information 
vote simultaneously in order to choose an alternative from a given set (typically, of size two).  As in our model, the alternatives have a state-dependent common value and the agents receive private information about the state.  According to the Condorcet Jury theorem, using the Law of Large Numbers, one can show that if voters vote naively (`truthfully') then the aggregated decision made under the majority rule is asymptotically optimal when the population of voters grows.
Austen-Smith and Banks \cite{Austen-Smith1996} challenge this premise by noting that naive voting is not necessarily rational (it does not form an equilibrium). Mclennan \cite{McLennan1998} provides an alternative framework where Condorcet's asymptotic efficiency results hold in equilibrium. In our setting we find that crowdfunding may lead to an inefficient outcome with a positive probability that is bounded away from zero, even when the population size goes to infinity. Thus we show that the crowdfunding mechanism does not satisfy efficiency as in Condorcet Jury Theorem.

The paper is organized as follows. In Section \ref{sec:model} we present the crowdfunding game and Theorem \ref{thm:unique_eq} which states that in this game we have a unique symmetric equilibrium. In Section \ref{sec:asympt} we examine the efficiency of the collective decision as the size of the population grows.
In Section \ref{sec:finite} we examine applicability of these results to finite populations.
Section \ref{sec:conc} concludes.

\section{The Crowdfunding Game}\label{sec:model}

A crowdfunding game, $\Gamma(B,n)$, is a game of incomplete information and common value with $n$ players. Let $\Omega=\{H,L\}$ denote the set of states of nature. Each state is drawn with probability $\frac{1}{2}$. Conditional on the realized state $\omega$, a private signal $s_i\in S_i=\{H,L\}$ is drawn independently for every player $i$. We assume $Pr(s_i=\omega|\omega)=p$ for some $p\in(\frac{1}{2},1)$. Each player $i$ has a binary action set, $A_i= \{0,1\}$,
with $a_i=1$ representing a decision to commit to buying the product if it is eventually supplied and $a_i=0$ represents a decision to opt-out and not to buy the product.  The utility of every player $i\in N$ is defined as follows
\begin{equation}\label{eq:consumer_util}
u_i(a_i,a_{-i},\omega)=\begin{cases}
1&\mbox{ if }a_i=1\mbox{ and }\sum_{j\in N} a_j \geq B\mbox{ and }\omega=H\\
-1&\mbox{ if }a_i=1\mbox{ and }\sum_{j\in N} a_j \geq B\mbox{ and }\omega=L\\
0&\mbox{ otherwise}
\end{cases}.
\end{equation}
In words, if player $i$ chooses to not to buy the product, she receives a utility of $0$. If she chooses to buy, then her utility is determined by the actions of the other players and the state of nature. If the number of other players who chose to buy is strictly lower than $B-1$, then the product will not be supplied and thus  player $i$'s utility will be $0$. If there are at least $B-1$ other players who choose to buy the product, then her utility is determined by the state of nature and equals $1$ in state  $H$ and $-1$ in state $L$.

A strategy for player $i$ is a mapping $\sigma_i:S_i\rightarrow \Delta A_i$. For simplicity we identify
$\sigma_i(s)$ with the probability that $i$ assigns to action $1$ conditional on signal $s$.

\subsection{Equilibrium}

A \emph{Bayes-Nash equilibrium} is a strategy profile $\sigma$ such that
$$E_{\sigma}(u_i(\sigma_i(s_i),\sigma_{-i}(s_{-i}))) \ge  E_{\sigma}(u_i(f(s_i),\sigma_{-i}(s_{-i})))
\ \  \forall i,\ \ \forall f:S_i\to A_i.$$

In the crowdfunding game, whenever $B>1$, there is a trivial equilibrium in which all players choose to opt-out. In order to avoid such equilibria we restricting attention to equilibria for which there is a positive probability that the good be supplied:
\begin{definition}
	A strategy profile (in particular an equilibrium strategy profile) $\sigma=(\sigma_1,\dots,\sigma_i,\dots,\sigma_{n})$ is called \textit{non-trivial} if, $$Pr_{\sigma}(\sum_{i} a_i \geq B)>0.$$
\end{definition}

We furthermore restrict attention to symmetric equilibria:
\begin{definition}
	A strategy profile (in particular an equilibrium strategy profile) is called \textit{symmetric} if there exists a strategy $\sigma^*$ such that $\sigma_i=\sigma^*$ for every player $i\in N$.
\end{definition}

We can now state our first result.
\begin{theorem}\label{thm:unique_eq}
	For every $n$ and every $B\in\{1\dots n\}$, there exists a \emph{unique} symmetric non-trivial Bayesian Nash equilibrium $\sigma=(\sigma_1,\ldots,\sigma_n)$ of $~\Gamma(B,n)$. Moreover, $\sigma_i$ has the following form,
	\begin{equation}\label{eq_equilibrium}
	\sigma_i(s_i)=\begin{cases}
	1&\mbox{ if }s_i=H\\
	\lambda=\lambda(B,n)\in[0,1) &\mbox{ if }s_i=L.
	\end{cases}.
	\end{equation}
\end{theorem}

We relegate the proof to Appendix \ref{sec:proofs} but provide some insights about this result.

Note that the posterior probability assigned to the state $H$ by a player receiving the signal $H$  is $p>\frac{1}{2}$, so she clearly takes the action $1$. On the other hand, the posterior assigned by a player with the low signal, $L$, is  $(1-p)<\frac{1}{2}$ and, naively speaking, such a player should prefer the action $0$. This is indeed the case when $B=1$. However, this intuition fails whenever $B$ is sufficiently large. Consider for example the case where $B>\frac{n}{2}$. Assume that all players of type $L$  always opt-out and that players of type $H$ always support and play $1$ A type $L$ player knows the good will be supplied only when the number of players of type $H$ exceeds $B$, in this case she can deduce that the probability of state $H$ is higher than $\frac{1}{2}$ conditional on the number of $H$ players is larger then $B$ even when taking into consideration her own signal. In that case she can achieve a positive expected utility by deviating and taking action $1$, whereas action $0$ will give her a utility of $0$. This demonstrates, that players' equilibrium strategy is not determined merely by their signal as players must condition their utility on the event that the good is supplied.

The above intuition suggests that the following probabilities play a crucial role in the analysis of symmetric equilibria:
\begin{definition}\label{def:xnyn}
We define the following probabilities,
	\begin{itemize}
		\item
		$x_n=Pr_{\sigma}(\sum_{i=1}^n a_i\geq B|\omega=H,a_1=1)$ is the probability the good is supplied conditional on state $\omega=H$ and on $a_i=1$.
		\item
		$y_n=Pr_{\sigma}(\sum_{i=1}^n a_i\geq B_n|\omega=L,a_1=1)$ is the probability that the good is supplied conditional on state $\omega=L$ and $a_1=1$.
	\end{itemize}
\end{definition}
By Theorem \ref{thm:unique_eq}, players of high type play the pure action $a=1$ in equilibria. Therefore, our analysis on the properties of the crowdfunding game depends on the equilibrium strategy for low-type players. In the following lemma we describe their expected utility from buying the product using the two probabilities defined above.
\begin{lemma}\label{lem:expct_util}
	The expected payoff of a low type $i$ (i.e., a player with $s_i=L$)
	from playing $a_i=1$ is
	$$ E_{\sigma}(u_i(1,\sigma_{-i}(s_{-i}))|s_i=L) = (1-p) x_n-py_n.$$
\end{lemma}
\begin{proof}
Note that such a player assigns a probability of $1-p$ to the event that the state of nature is $\omega=H$  and a probability of $p$ to the event that $\omega=L$.
In the former case, if she plays $a_i=1$, then her utility is $1$ provided that the product is supplied, this holds with probability $x_n$. In the latter case, her utility from playing $a_i=1$ is $-1$, whenever the good is supplied, an event which probability is $y_n$.
\end{proof}

A low type plays a mixed action only when he is indifferent between the two pure actions and so a necessary condition for $\lambda>0$ (recall it is always the case that $\lambda<1$) is: $(1-p) x_n-py_n=0.$ This result becomes useful when analyzing the behavior of the players as the size of the population increases.
\section{Asymptotic Analysis}\label{sec:assympt_ana}
As mentioned in the introduction, firms turn to utilize the crowdfunding scheme for two distinct objectives. The first objective is to estimate the market demand for their product, and the second is to signal quality for future investors. The difference between these objectives is that in the former, it is in the firm's best interest to shelve the product if it is bad, while in the latter, the firm is less concerned with the true value of the product and more on the success of the campaign as the assumed risks are to be divided to future investors. To gain a complete view over the efficiency of the crowdfunding game we suggest two indexes. The first is the \emph{correction index} which refers to the efficiency in which this game filters projects and the second is the \emph{market penetration index} which refers to the expected number of players. We study the performance of the crowdfunding game as the size of the population increases.

\subsection{Asymptotic Correctness}\label{sec:asympt}

One important criterion to evaluate the success of a crowdfunding campaign is how well it aggregates information. In this spirit we introduce the {\em {correctness index}} for the crowdfunding game. The correctness index is the probability that the players receive the correct decision in hindsight. Formally,

\begin{definition}
	The \textit{correctness} parameter $\theta(B,n)$ of a game $\Gamma(B,n)$ is:
	\begin{equation}
	\theta(B,n)=\frac{1}{2}Pr(\sum_{i\in N} a_i \geq B|\omega=H)+\frac{1}{2}Pr(\sum_{i\in N} a_i < B|\omega=L)
	\end{equation}
	where $\sigma$ is the unique symmetric non-trivial equilibrium.
\end{definition}

Our second result characterizes the asymptotic correctness of the crowdfunding game,
\begin{theorem}\label{thm:asym_correctness}
	\begin{equation}\label{eq:max}
	\lim_{n\rightarrow\infty} \max_{B\in\{1\dots n\}} \theta(B,n) =\frac{3p-1}{2p} \ \ (<1).
	\end{equation}
\end{theorem}

The correctness in the crowdfunding game is bounded away from one. This stands in contrast with Condorcet's jury theorem which argues that where large societies necessarily vote for the correct alternative.

The full proof of Theorem \ref{thm:asym_correctness} is relegated to Appendix \ref{sec:proofs} however we provide an outline for the proof. The proof requires the following two lemmas (which proofs are provided in Appendix \ref{sec:proofs} as well). The first lemma characterizes the asymptotic probability that a player of type $L$ chooses the action $1$:

\begin{lemma}\label{lem:asympt_probs}
	For a sequence $\{B_n\}_{n=1}^{\infty}$ satisfying $\lim_{n\rightarrow\infty} \frac{B_n}{n}=q$ for some $q\in[0,1]$:
	\begin{equation}
	\lim_{n\rightarrow\infty} \lambda(B_n,n)=\begin{cases}
	0&\mbox{ if } q\leq 1-p\\
	\frac{q-(1-p)}{p}&\mbox{ otherwise}
	\end{cases}
	\end{equation}
\end{lemma}

Lemma \ref{lem:diff_is_zero} is a corollary of Lemma \ref{lem:asympt_probs},
\begin{lemma}\label{lem:diff_is_zero}
	For a sequence $\{B_n\}_{n=1}^{\infty}$ satisfying $\lim_{n\rightarrow\infty} \frac{B_n}{n}=q$ for some $q\in[0,1]$:
	$$\lim_{n\rightarrow\infty}x_n-Pr(\sum_{i\in N} a_i \geq B_n|\omega=H)=0$$
	and
	$$\lim_{n\rightarrow\infty}y_n- Pr(\sum_{i\in N} a_i \geq B_n|\omega=L)=0.$$
\end{lemma}
By Lemma \ref{lem:diff_is_zero} we can conclude that asymptotically, player $i$ is conditionally non-pivotal. In other words, from the perspective of player $i$, if the population is very big, she consider her own action to have no role in determining the probability of supply, conditional on knowing the state of nature.
We are now ready to outline the proof of Theorem \ref{thm:asym_correctness}.

\subsubsection*{Proof Outline of Theorem \ref{thm:asym_correctness}}
Consider the sequence of games $\Gamma(\frac{n}{2},n)$.
By Lemma \ref{lem:asympt_probs}
\begin{equation}\label{eq:assym_lambda}
\lambda(\frac{n}{2},n) = \frac{2p-1}{2p}.
\end{equation}
By equation \eqref{eq:assym_lambda}, for  sufficiently large $n$, $\lambda(\frac{n}{2},n)>0$, and thus by the indifference condition for player of type $l$,
$(1-p)x_n-py_n=0$.

Further more, by equation \eqref{eq:assym_lambda}, $\lambda_H(\lambda(\frac{n}{2},n))$ converges to $\frac{3p-1}{2p}>\frac{1}{2}.$ This implies that $\lim_{n\rightarrow\infty} x_n =1$, and thus $\lim y_n = \frac{1-p}{p}.$  Together with Lemma \ref{lem:diff_is_zero} this proves that
$$\lim_{n\rightarrow\infty}\Theta(\frac{n}{2},n)=\frac{1}{2}+\frac{1}{2}(1-\frac{1-p}{p})=\frac{3p-1}{2p}.$$
And by this we conclude that,
$$\lim_{n\rightarrow\infty} \max_{B\in\{1\dots n+1\}} \theta(B,n) \ge \frac{3p-1}{2p}.$$

To show the opposite inequality consider an arbitrary sequence $\{B_n\}$ for which the sequences
$\{\theta(B_n,n)\}, \{x_n\}_n$ and $\{y_n\}_n$ converge to $\theta^*, x^*$ and $y^*$ respectively (otherwise, consider a sub-sequence). By the definition of the correctness index and by Lemma \ref{lem:asympt_probs} we get,
\begin{equation}\label{eq:correct_ll1}
\lim_n \theta(B_n,n) = \frac{1}{2}x^* + \frac{1}{2}(1-y^*).
\end{equation}

Furthermore assume that $q>1-p$ and thus for sufficiently large $n$,  $\lambda(B_n,n)>0$. This entails that
$(1-p)x_n-py_n=0$ and consequently
\begin{equation}\label{eq:correct_ll2}
(1-p)x^*-py^*=0.
\end{equation}

Therefore, by equation \eqref{eq:correct_ll1} and equation \eqref{eq:correct_ll2},  the asymptotic correctness value is bounded above by the solution for the following linear program:
\begin{equation}
\begin{aligned}
& {\text{max}}
& & \frac{1}{2}x^*+\frac{1}{2}(1-y^*) \\
& \text{s.t.} & & 1\geq x^*,y^*\geq 0 \\
& & &  (1-p)x^*-py^*=0, \\
\end{aligned}
\end{equation}
which is $\frac{3p-1}{2p}$.
\qed

\subsection{Market Penetration}

An alternative motivation for running a crowdfunding campaign is to use it to convince some third party (e.g., institutional investors) about the validity of the product.  To capture this we introduce the {\em market penetration index} which is the expected share of the population who eventually buys the product (in case the product is not supplied the number of buyers is zero). Formally,
\begin{definition}
	The {\em market penetration index} of the game $\Gamma(B,n)$ is
	\begin{equation}
	R(B,n)=E_{\sigma}\big[\frac{(\sum_{k=1}^{n} a_i)}{n} \chi(\sum_{i\in N}a_i \geq B) \big] =
	Pr_{\sigma}(\sum_{i\in N}a_i \geq B)E_{\sigma}\big[\frac{\sum_{i\in N}a_i}{n}|\sum_{i\in N}a_i \geq B \big],
	\end{equation}
	where $\chi(A)$ is the indicator function of the event $A$.
\end{definition}

The market penetration index can also be motivated is in terms of revenue. To see this consider an alternative but equivalent model. In this model the good is valued either at $2$ (in state $H$) or  $0$ (in state $L$). Assume there are no marginal production costs and players, with linear utilities, can join the campaign by pledging $1$ unit of currency. In this equivalent setting $R(B,n)$ is merely the firm's expected per-capita profit. The following theorem characterizes the asymptotic maximal market penetration.

\begin{theorem}\label{thm:revenue}
	\begin{equation}\label{eq:revenue}
	\lim_{n\rightarrow\infty} \max_{B\in\{1\dots n\}} R(B,n) =\frac{1}{2p}
	\end{equation}
\end{theorem}

A counter-intuitive conclusion is that the penetration level decreases as the signal, $p$, becomes more accurate.
To understand the mechanism that leads up to this result note that as  the signal becomes less informative, the players with the low-type will be less confidant in their personal signal and more inclined to use the aforementioned social insurance. This can be seen in equation \eqref{eq:assym_lambda}. For a sufficently large $n$,  the probability that low-type player assigns to buying increases in  $p$ increases. As $\lambda$ increases and thus so does the expected number of low-type player who choose `buy'. Furthermore, this leads to an increase in  the probability a bad product will be supplied increases and thus again, increasing $R(B,n)$.

Lemmas \ref{lem:asympt_probs} and \ref{lem:diff_is_zero} are, once more, instrumental for the proof of Theorem \ref{thm:revenue} as we explain below. Once again, the full proof is relegated to Appendix \ref{sec:proofs}.
\subsubsection*{Proof Outline of Theorem \ref{thm:revenue}}
We consider a sequence games $\Gamma(q n,n)$ for an arbitrary $q\in((1-p),1).$
As before, assume that the sequences $\lambda(q n,n),  x_n$ and $y_n$ converge (if not, consider a sub-sequence).

By Lemma \ref{lem:asympt_probs} we deduce that $\lim_{n\rightarrow\infty} \lambda(q n,n) = \frac{q-(1-p)}{p}.$
Therefore $\lambda_H(\lambda(q n,n))$ converges to $\frac{q(1-p)+2p-1}{p}>q$, where the last inequality follows from the choice of $q>1-p$. In words, conditional on the state being $H$, the limit probability that each player takes the action $1$ is greater than $q$. Thus, by the Law of Large Numbers, the limit probability of the event in which the number of players is greater then the threshold $qn$ approaches $1$.  By applying Lemma \ref{lem:diff_is_zero} we conclude that $\lim_{n\rightarrow\infty} x_n =1$. Thus, for a sufficiently large $n$, $\lambda(q n,n)>0$.  By the indifference criterion of a low type players, we have that
$(1-p)x_n-py_n=0$ and therefore, $\lim_{n\rightarrow\infty} y_n = \frac{1-p}{p}.$

Applying Lemma \ref{lem:diff_is_zero} again we conclude that the limit probability to pass the threshold in the state $L$ is $\frac{1-p}{p}.$ We  compute the limit market penetration index by calculating its expected value over the two states:
$$\lim_{n\rightarrow\infty}R(qn,n)\geq \frac{1}{2}q+\frac{1}{2}\frac{1-p}{p}q=\frac{q}{2p}.$$
By taking $q$ to one we conclude that  $lim_{n\rightarrow\infty} \max_{B\in\{1\dots n\}} R(B,n) \ge \frac{1}{2p}$

To verify the opposite inequality  consider a sequence of games $\Gamma(B_n,n)$ for which the maximum of equation \eqref{eq:revenue} is attained. Assume that the sequences $\{\lambda_n\}, \{x_n\},$ and $\{y_n\}$ converge to $\lambda^*,x^*$ and $y^*$ respectively.  As the expected penetration index cannot exceed $1$, the definition of the market penetration index yields
$R(B_n,n) \le \frac{1}{2}x_n+\frac{1}{2}y_n$. In the limit we get
$$\lim_{n\rightarrow\infty} \max_{B\in\{1\dots n\}} R(B,n)\leq \frac{1}{2}x^*+\frac{1}{2}y^*.$$

By the indifference criterion for the low-type player, which holds for large enough $n$ we get that $(1-p)x_n-py_n=0$. This must also hold in the limit  and thus $(1-p)x^*-py^*=0$. Therefore the maximal penetration is bounded from above by the solution of the following linear program:
\begin{equation}
\begin{aligned}
& {\text{max}}
& & \frac{1}{2}x^*+\frac{1}{2}y^* \\
& \text{s.t.} & & 1\geq x^*,y^*\geq 0 \\
& & &  (1-p)x^*-py^*=0, \\
\end{aligned}
\end{equation}
which is $\frac{1}{2p}$.
\qed

\section{Finite number of players}\label{sec:finite}
In Section \ref{sec:assympt_ana} we examine the properties of the crowdfunding game in a setting where the number of players is very large.  Using our proposed correctness index we show that there is a positive probability that the campaign's results will not be optimal, and using the market penetration index we show that firms can utilize this game to attract low type players thus preventing the information aggregation.

When looking at figures from crowdfunding platforms one can see that the number of players is usually rather small. For example, the average number of pledgers on Kickstarter is $108$, on other platforms, this figure is even lower (see \url{https://www.entrepreneur.com/article/269663}).  To bridge this gap, and examine the applicability of our results for small samples,  we calculate the equilibrium strategy $(\lambda)$, the correctness index $\left(\theta(B,n)\right)$  and the market penetration index $\left(\frac{R(B_n,n)}{n}\right)$ for various values of the game parameters. The results of these calculations are found in Appendix \ref{sec:finite_tab}, Table \ref{tab:tab1}.

The most interesting observation is that our asymptotic analysis provides a good approximation for crowdfunding games with interim values of $n$. In some case even as values as low as $n=10$. This is quite robust to the game parameters $B$ and $p$ and is true for all three parameters we calculate: $\lambda, \Theta$ and $R$.

We additionally observe that when the signal is weak ($p=0.55$) and the threshold is low ($B=\frac{n}{3}$), the low-type players always opt-out. Nevertheless, the number of high-type players is sufficient to induce production even if $\omega=L$ and thus the correctness of the game deteriorates rather rapidly. As $B$ increases the risk facing low-type players decreases and therefore we can see that throughout the table, higher $B$ leads to a higher $\lambda$. We can also see that the asymptotic results provide a very good approximation  even in cases with a very small population such as $n\in\{5,10\}$. One exception is the case where $n=5$ and $B=0.9n$. In these cases the product will only be supplied if all players choose to commit. Note that even then, we still get a good approximation when $n=100$ or higher.

\section{Concluding remarks}\label{sec:conc}
Crowdfunding is often used by many entrepreneurs to validate the market demand for innovative products or an art project. We study how well do crowdfunding campaigns perform in this context. To do so we introduce a vary simple game of incomplete information we call the crowdfunding game. For a crowdfunding game we consider two success measures. First, the `correctness' of a campaign which captures how well information is aggregated and second, the market penetration index that reflects how convincing the campaign is. We show that for large populations information is not fully aggregated, which stands in contrast with Condorcet's jury theorem and we also provide clear bounds on the correctness and penetration index.

Our results are primarily asymptotic. However, calculations show that these asymptotic bounds provide good approximations for realistic values of populations size, sometimes as small $10$ players. In fact, even when the number of players is finite, all three parameters we measure: $\lambda,$ the correctness index and the market penetration index, are quite close to their respective asymptotic values. This observation is robust to the game parameters $B$ and $p$.

Another aspect that is ignored is the possibility to offer differentiated products at varying prices, which is quite typical in crowdfunding campaigns. Finally, actual campaigns are naturally dynamic and so it is natural to think of a scenario where buyers can wait for others, possibly more informed buyers, to make their move. These  scenarios are missing from our static model. We wish to study the interplay between social insurance and information aggregation. For this purpose, this model is sophisticated enough to yield interesting insights on the crowdfunding phenomenon which result from the strategic interaction among buyers.

\bibliographystyle{plain}
\bibliography{crowdfunding}

\begin{thebibliography}{10}

\bibitem{Alaei2016}
Saeed Alaei, Azarakhsh Malekian, and Mohamed Mostagir.
\newblock {Working Paper A Dynamic Model of Crowdfunding}.
\newblock In {\em Proceedings of the 2016 ACM Conference on Economics and
  Computation}, Maastricht, The Netherlands, 2016. ACM.

\bibitem{Austen-Smith1996}
David Austen-Smith and Jeffrey~S Banks.
\newblock {Information Aggregation , Rationality , and the Condorcet Jury
  Theorem}.
\newblock {\em The American Political Science Review}, 90(1):34--45, 1996.

\bibitem{Chemla2016}
Gilles Chemla and Katrin Tinn.
\newblock {Learning Through Crowdfunding}.
\newblock 2016.

\bibitem{Ellman2014}
Matthew Ellman and Sjaak Hurkens.
\newblock {Optimal Crowdfunding Design}.
\newblock {\em NET Institute Working Paper No. 14-21}, (September), 2014.

\bibitem{Feller1968}
William Feller.
\newblock {\em {An introduction to probability theory and its applications:
  volume I}}, volume~3.
\newblock John Wiley {\&} Sons New York, 1968.

\bibitem{Hemer2011}
Joachim Hemer.
\newblock {A snapshot on crowdfunding}.
\newblock {\em Working Paper}, 2011.

\bibitem{Kumar2017}
Praveen Kumar, Nisan Langberg, and David Zvilichovsky.
\newblock {( Crowd ) funding Innovation , Financing Constraints and Real E ¤
  ects ( Crowd ) funding Innovation , Financing Constraints and Real E ¤
  ects}.
\newblock pages 1--40, 2017.

\bibitem{Lei2017}
Yu~Lei, Ali~Alper Yayla, and Surinder~Singh Kahai.
\newblock {Guiding the Herd : The Effect of Reference Groups in Crowdfunding
  Decision Making}.
\newblock In {\em Proceedings of the 50th Hawaii International Conference on
  System Sciences}, pages 1912--1921, 2017.

\bibitem{McLennan1998}
Andrew Mclennan.
\newblock {Consequences of the Condorcet Jury Theorem for Beneficial
  Information Aggregation by Rational Agents}.
\newblock {\em Source: The American Political Science Review}, 92(2):413--418,
  1998.

\bibitem{Mollick2014}
Ethan~R. Mollick.
\newblock {The dynamics of crowdfunding: An exploratory study}.
\newblock {\em Journal of Business Venturing}, 29(1):1--16, 2014.

\bibitem{Mollick2015}
Ethan~R. Mollick.
\newblock {Delivery Rates on Kickstarter}.
\newblock 2015.

\bibitem{Strausz2017}
Roland Strausz.
\newblock {A Theory of Crowdfunding - a mechanism design approach with demand
  uncertainty and moral hazard}.
\newblock {\em American Economic Review}, 107(6):1--40, 2017.

\bibitem{Yang2016}
Yang Yang, Harry~Jiannan Wang, and Gang Wang.
\newblock {Understanding Crowdfunding Processes : a Dynamic Evaluation and
  Simulation Approach}.
\newblock {\em Journal of Electronic Commerce Research}, 17(1):47--64, 2016.

\bibitem{Yum2012}
Haewon Yum, Byungtae Lee, and Myungsin Chae.
\newblock {From the wisdom of crowds to my own judgment in microfinance through
  online peer-to-peer lending platforms}.
\newblock {\em Electronic Commerce Research and Applications}, 11(5):469--483,
  2012.

\end{thebibliography}
\pagebreak
\appendix
\section{Table of Finite Sample Calculations.}
\label{sec:finite_tab}
\begin{table}[!htbp]
	\centering
	\begin{tabular}{|c|c|c|c|c|c|c|c|c|c|c|c|c|c|}
		\hline
		$p$&$n$&\multicolumn{4}{|c|}{$B=\lceil\frac{n}{3}\rceil$}&\multicolumn{4}{|c|}{$B=\lceil\frac{n}{2}\rceil$}&\multicolumn{4}{|c|}{$B=\lceil\frac{9n}{10}\rceil$}\\
		\hline
		&&$\lambda$&\small $\theta(B,n)$&$\frac{x_n+y_n}{2}$&$\frac{R(B,n)}{n}$&$\lambda$&$\theta(B,n)$&$\frac{x_n+y_n}{2}$&$\frac{R(B,n)}{n}$&$\lambda$&$\theta(B,n)$&$\frac{x_n+y_n}{2}$&$\frac{R(B,n)}{n}$\\
		\hline
		\multirow{6}{*}{$.55$}&$5$&$0$&$.562$&$.806$&$.468$&$.029$&$.590$&$.526$&$.366$&$.599$&$.541$&$.329$&$.329$\\
		&$10$&$0$&$.582$&$.816$&$.453$&$.064$&$.605$&$.690$&$.427$&$.672$&$.558$&$.495$&$.462$\\
		&$100$&$0$&$.505$&$.995$&$.498$&$.154$&$.603$&$.891$&$.526$&$.837$&$.584$&$.801$&$.745$\\
		&$1000$&$0$&$.5$&$1$&$.5$&$.115$&$.595$&$.904$&$.511$&$.832$&$.592$&$.905$&$.832$\\
		&$\infty$&$0$&$.5$&$1$&$.5$&$.091$&$.591$&$.909$&$.909$&$.818$&$.591$&$.909$&$.909$\\
		\hline
		
		\multirow{6}{*}{$.75$}&$5$&$0$&$.809$&$.676$&$.459$&$.044$&$.883$&$.525$&$.408$&$.571$&$.712$&$.355$&$.355$\\
		&$10$&$0$&$.886$&$.610$&$.424$&$.101$&$.895$&$.593$&$.439$&$.651$&$.783$&$.501$&$.474$\\
		&$100$&$.078$&$.859$&$.641$&$.436$&$.291$&$.852$&$.648$&$.489$&$.838$&$.839$&$.659$&$.625$\\
		&$1000$&$.102$&$.841$&$.659$&$.442$&$.323$&$.839$&$.661$&$.497$&$.860$&$.835$&$.665$&$.632$\\
		&$\infty$&$.111$&$.833$&$.667$&$.667$&$.333$&$.833$&$.667$&$.667$&$.867$&$.833$&$.667$&$.667$\\
		\hline
		\multirow{6}{*}{$.95$}&$5$&$0$&$.989$&$.511$&$.480$&$.053$&$.995$&$.504$&$.479$&$.437$&$.923$&$.444$&$.444$\\
		&$10$&$.052$&$.994$&$.506$&$.479$&$.128$&$.991$&$.509$&$.483$&$.551$&$.974$&$.506$&$.496$\\
		&$100$&$.217$&$.981$&$.519$&$.487$&$.378$&$.979$&$.521$&$.495$&$.827$&$.976$&$.524$&$.518$\\
		&$1000$&$.273$&$.976$&$.524$&$.490$&$.446$&$.975$&$.525$&$.499$&$.877$&$.974$&$.526$&$.520$\\
		&$\infty$&$.298$&$.974$&$.526$&$.526$&$.474$&$.974$&$.526$&$.526$&$.895$&$.974$&$.526$&$.526$\\
		\hline
		
	\end{tabular}
	\caption{Calculation results for finite players sets}\label{tab:tab1}
	\raggedright{
		A crowdfunding game is defined by three parameters: (1) The number of players ($n$), (2) The threshold $B$, which in this analysis be described as a fraction of $n$, and (3) the quality of the player signals $p$. The results of the model are: (1) the probability that a low-type player will buy in equilibrium ($\lambda$), (2) The correctness of the game and (3) the expected number of contributions, provided that the product is supplied (MPI). In the following table we present the results for: different signal qualities $p\in\{0.55,0.75.0.95\}$, various thresholds $B\in\{\frac{n}{3},\frac{n}{2},0.9n\}$ and an increasing finite number of players $n=\{5,10,100,1000\}$, where the rows in which $n=\infty$ hold the value of the results at limit probabilities calculated using Theorems \ref{thm:unique_eq}, \ref{thm:asym_correctness} and \ref{thm:revenue}.
	}
\end{table}
\FloatBarrier
\section{Proofs}\label{sec:proofs}

Assume $\sigma$ is strategy tuple  of $\Gamma(B,n)$. Let $\sigma_i(H)=\psi$ and $\sigma_i(L)=\lambda$ be the probabilities that a high type player and the low type player choose action $a_i=1$ respectively.
Conditional on the state of nature $\omega\in\{H,L\}$ the random variable $\sum_{i=1}^n a_i$ has binomial distribution $Bin(n,\lambda_{\omega})$, where $\lambda_H=p\psi+(1-p)\lambda$ and $\lambda_L=(1-p)\psi+p\lambda$.

Let $x=Pr_{\sigma}(\sum_{j\neq i} a_i \geq B-1|\omega=H,a_i=1)$ be the  probability that the product be supplied conditional on state $\omega=H$ and $a_i=1$. Similarly, let $y=Pr_{\sigma}(\sum_{j\neq i} a_i \geq B-1|\omega=L,a_i=1)$ be the that the product be supplied probability conditional on $\omega=L$ and $a_i=1$.

The expected payoff to player of type $L$ conditional on her taking the action $1$ can be derived by conditioning on the state of nature. Following the definition of $x$ and $y$ this payoff must equal $(1-p)x-py$. Similarly, the payoff to player of type $H$ conditional on taking the action $1$ is $px-(1-p)y$.

\subsection{Proof of Theorem \ref{thm:unique_eq}}

In the following lemma we characterize the properties of  a symmetric non-trivial Bayes-Nash equilibria in all $\Gamma(B,n)$,

\begin{lemma}\label{lem:H}
	If $\sigma$ is a symmetric non-trivial Bayes-Nash equilibrium of $\Gamma(B,n)$ then $x>y$ and $\psi=1 > \lambda \ge 0$.
\end{lemma}

In words the lemma states that in every symmetric non-trivial Bayesian Nash equilibrium the high type player chooses $a_i=1$ with probability $1$, and conditional on player $i$ choosing $a_i=1$, the probability that the product will be supplied is strictly higher in state $\omega=H$ compared with state $\omega=L.$

\begin{proof}
	
	We show first that $x>y$. Assume by contradiction that there exists a non-trivial symmetric equilibrium $\sigma$ in which $y\geq x$. In this case as $p>1/2$, the expected payoff of a low type player from playing $a_i=1$ is $(1-p)x-py<0$. Therefore, under $\sigma$ we have that $\lambda=0$. Since $\sigma$ is a non trivial, symmetric equilibrium of $\Gamma(B,n)$, it must be the case that $\psi>0$.  However when $\lambda=0$ and $\psi>0$ we get that $\lambda_H>\lambda_L$. This in turn implies that $x>y$ which establishes a contradiction.
	
	Since $x>y$ and $p>\frac{1}{2}$ the expected payoff of high type player from taking action $1$ is $px-(1-p)y >0$. On the other hand, by playing $a_i=0$ she gains a payoff of $0$ and thus, high type player will always play $\psi=1$. Finally, note that if $\lambda=1$ then $x=y=1$ which contradicts the inequality $x>y$.
\end{proof}

Let $\sigma^\lambda$ be the strategy profile in which a low type player plays action $1$ with probability $\lambda$ and a high type player plays action $1$ with probability one. By Lemma \ref{lem:H} we know that the equilibrium strategy must be of this form for some $\lambda\in [0,1)$.

We need to show that $\sigma^\lambda$ is an equilibrium strategy for a unique $\lambda$. 
Let $x^\lambda$ and $y^\lambda$ be the probabilities that the threshold $B$ is reached conditional on the state $\omega$ and on $a_i=1$ for $\omega\in H,L$ respectively. In the following proposition we prove that there is at most one value of $\lambda$ in which the indifference condition of the a low type player equals zero. 

\begin{proposition}\label{prop:auxth}
	If $(1-p)x^{\lambda'}-py^{\lambda'}=0$, for some $\lambda'\in(0,1)$  then  $(1-p)x^{\lambda'}-py^{\lambda'}>0$ for all $\lambda\in[0,\lambda')$ and  $(1-p)x^{\lambda'}-py^{\lambda'}<0$ for all $\lambda\in[\lambda',1].$
\end{proposition}

\begin{proof}[\textbf{Proof of Proposition \ref{prop:auxth}}]
	
	Recall that
	$x^\lambda=Pr_{\sigma^\lambda}(\sum_{j\neq i} a_i \geq B-1|\omega=H,a_i=1)$ and
	$\sum_{j\neq i} a_i \sim Bin(n-1,\lambda_H(\lambda))$,
	where $\lambda_H(\lambda)=p+(1-p)\lambda$.
	Similarly,
	$y^\lambda=Pr_{\sigma^\lambda}(\sum_{j\neq i} a_i \geq B-1|\omega=L,a_i=1)$, and $\sum_{j\neq i} a_i \sim Bin(n-1,\lambda_L (\lambda))$ where $\lambda_L(\lambda)=p\lambda+(1-p)$.
	
	For every $\gamma\in[0,1]$,  we define an auxiliary random variable $z(\gamma)\sim Bin(n-1,\gamma)$.  Let $\varphi(\gamma)= Pr(z(\gamma)\geq B-1)$ and thus we can write $x^\lambda=\varphi(\lambda_H(\lambda))$ and
	$y^\lambda=\varphi(\lambda_L(\lambda)).$
	
	We let $g(\lambda)$ be the expected payoff of a low type player $i$ from playing $a_i=1$, when the other  players follow the strategy $\sigma^\lambda$. By definition,
	\begin{equation}\label{eq:g1}
	g(\lambda)=(1-p)x^\lambda-py^\lambda=(1-p)\varphi(\lambda_H(\lambda))-p\varphi(\lambda_L(\lambda)).
	\end{equation}
	We calculate the derivative of equation \eqref{eq:g1},
	\begin{equation}\label{eq:deriv_g}
	\begin{split}
	&g'(\lambda)=(1-p)\varphi'(\lambda_H(\lambda))\lambda_H'(\lambda)-p\varphi'(\lambda_L(\lambda))\lambda_L'(\lambda)=\\
	&=(1-p)^2 \varphi'(\lambda_H(\lambda))-p^2\varphi'(\lambda_L(\lambda))
	.\end{split}
	\end{equation}
	
	To prove the proposition we shall show that if there exist $\lambda'\in[0,1)$ such that $g(\lambda')=0$, then  $g'(\lambda')<0$.
	This implies first that there exists at most one such $\lambda'$ and second, if such $\lambda'$ exists, then it holds that $g(\lambda)<0$ for every $\lambda>\lambda'$ and $g(\lambda)>0$
	for every $\lambda<\lambda'$ and thus completes the proof of the proposition.
	
	Assume by contradiction that there exists some $\lambda'\in(0,1)$ such that $g(\lambda')=0$ and $g'(\lambda')\geq 0$.
	First we calculate $\varphi'(\lambda)$. By \cite{Feller1968},
	\begin{equation}\label{eq:varphi}
	\varphi(\lambda)=(n-1)\binom{n-2}{B-2}\int_0^{\lambda} t^{B-2}(1-t)^{n-B}dt.
	\end{equation}
	Therefore,
	\begin{equation}\label{eq:varphi_deriv}
	\varphi'(\lambda)=(n-1)\binom{n-2}{B-2}\lambda^{B-2}(1-\lambda)^{n-B}
	\end{equation}
	and thus $\varphi'(\lambda)>0$ for all $\lambda\in(0,1)$.

	We assume $g'(\lambda')\geq 0$ and thus,
	\begin{equation}\label{eq:g_tag_1}
	g'(\lambda')=(1-p)^2 \varphi'(\lambda_H(\lambda'))-p^2\varphi'(\lambda_L(\lambda'))\geq 0.
	\end{equation}
	Since $\frac{1}{2}<p<1$ and $g'(\lambda')>0$ we can divide both sides of equation \eqref{eq:g_tag_1} by $p$ and get
	\begin{equation}\label{eq:g_tag_1}
	0\leq g'(\lambda')\leq \frac{(1-p)}{p}(1-p) \varphi'(\lambda_H(\lambda'))-p\varphi'(\lambda_L(\lambda')).
	\end{equation}
	As $\frac{1-p}{p}<1$ and $\varphi'(\lambda')>0$ the following inequality holds,
	\begin{equation}\label{eq:dirv_g_ubound}
	0\leq g'(\lambda')< (1-p) \varphi'(\lambda_H(\lambda'))-p\varphi'(\lambda_L(\lambda')).
	\end{equation}
	
	By equation \eqref{eq:varphi_deriv}, equation \eqref{eq:dirv_g_ubound} becomes,
	\begin{equation}\label{eq:deriv_g1}
	\begin{split}
	&0\leq g'(\lambda')<
	(n-1)\binom{n-2}{B-2}\left((1-p)\lambda_H(\lambda')^{B-2}(1-\lambda_H(\lambda'))^{n-B}-p\lambda_L(\lambda')^{B-2}(1-\lambda_L(\lambda'))^{n-B}\right)<\\
	&<C\eta(n-1,B-2)
	\end{split}
	\end{equation}
	Where $C=\frac{1}{1-\lambda_H(\lambda')}(n-1)\binom{n-2}{B-2}$ and $$\eta(k,n-1)=(1-p)\lambda_H(\lambda')^{k}(1-\lambda_H(\lambda'))^{n-1-k}-p \lambda_L(\lambda')^{k}(1-\lambda_L(\lambda'))^{n-1-k}.$$
	Therefore $\eta(n-1,B-2)>0$.
	
	Next we claim that if $\eta(k,n-1)>0$ it must be the case that $\eta(k+1,n-1)>0$ as well,
	\begin{equation}\label{eq:16}
	\begin{split}
	&\eta(k+1,n-1)=(1-p)\lambda_H(\lambda')^{k+1}(1-\lambda_H(\lambda'))^{n-1-(k+1)}-p\lambda_L(\lambda')^{k+1}(1-\lambda_L(\lambda'))^{n-1-(k+1)}=\\
	&\frac{\lambda_H(\lambda')}{1-\lambda_H(\lambda')}(1-p)\lambda_H(\lambda')^k(1-\lambda_H(\lambda'))^{n-1-k}-\frac{\lambda_L(\lambda')}{1-\lambda_L(\lambda')} p \lambda_L(\lambda')^k(1-\lambda_L(\lambda'))^{n-1-k} \geq \\ &\geq\frac{\lambda_L(\lambda')}{1-\lambda_L(\lambda')}\eta(k,n-1)
	.\end{split}
	\end{equation}
	Where the last inequality holds as $\lambda_H(\lambda')>\lambda_L(\lambda')$ and  due the function $\frac{x}{1-x}$ increases monotonically  for $x\in(0,1)$.
	
	The assumption $g'(\lambda')\geq 0$ entails that $\eta(B-2,n-1)>0$ and thus by equation \eqref{eq:16},
	$$
	g(\lambda')=\sum_{k=B-1}^{n} \binom{n-1}{k} \eta(k,n-1)>
	\sum_{k=B-1}^{n} \binom{n-1}{k} \eta(B-2,n-1)>0,
	$$
	in contradiction to $g(\lambda')=0$.
\end{proof}

By Proposition \ref{prop:auxth}, for every $B,n$, there can be at most one value for which the low type player is indifferent between playing the action $1$ and the action $0$ in  $\Gamma(B,n)$. The aforementioned $\lambda'$ is a good candidate for the probability that a low type player will assign to the action $1$ in a symmetric, non-trivial equilibrium.
We now turn to complete the proof of Theorem \ref{thm:unique_eq}. We  show that for every $B,n$, in all symmetric, non-trivial equilibria of $\Gamma(B,n)$, the low type player either plays the pure action $0$ or plays a mixed strategy in which she assigns a probability of $\lambda'$ to the action $1$. Furthermore we show that these cases are mutually exclusive and thus there is exactly one such equilibrium in $\Gamma(B,n)$.

\begin{proof}[\textbf{Proof of Theorem \ref{thm:unique_eq}}]
	We split the proof into two cases.
	First assume that  $(1-p)x^\lambda-py^\lambda\neq 0$ for every $\lambda$.  The continuity of $x^\lambda$ and $y^\lambda$ establishes that for every $\lambda$, either $(1-p)x^\lambda-py^\lambda< 0$ or $(1-p)x^\lambda-py^\lambda>0$. For $\lambda=1$ and every $B,n$ we have that  $x^1=y^1=1$. The expected utility of a low type player is therefore $(1-p)x^1-py^1=1-2p<0$, implying that the former inequality holds, and so the payoff of the low type player when taking action $1$ is always negative. We can therefore conclude that $\sigma^0$ is indeed the unique equilibrium.

	In the complementary case, there exists $\lambda'\in(0,1)$ such that  $(1-p)x^{\lambda'}-py^{\lambda'}=0$. In $\sigma^{\lambda'}$ the high type player takes the pure  action $1$ and receives a non-negative utility. Therefore, she cannot gain from deviating. On the other hand, the equality $(1-p)x^{\lambda'}-py^{\lambda'}=0$ implies that the low type player mixes between the two available pure actions and thus is indifferent between them. Therefore she cannot profitably deviate as well. This establishes the fact that $\sigma^{\lambda'}$ is indeed an equilibrium.
	
	To establish uniqueness we note that whenever  $\lambda<\lambda'$  the low type player receives positive expected payoff from taking the pure action $1$ in the strategy profile $\sigma^\lambda$ and can profitable deviate to the pure action $1$. Similarly whenever $\lambda>\lambda'$ the low type player receives a negative payoff in the strategy profile $\sigma^\lambda$ and so can  profitably deviate to the pure action $0$.
\end{proof}

\subsection{Proof of Lemma \ref{lem:asympt_probs}}

Let $q\in[0,1]$, and let $\{B_n\}_{n=1}^\infty$ be a sequence satisfying $\lim_{n\rightarrow\infty}\frac{B_n}{n}=q.$ Consider the corresponding sequence of games $\{\Gamma(B_n,n)\}_{n=1}^\infty$ .

By Theorem \ref{thm:unique_eq}, for every $\Gamma(B_n,n)$ there exists a unique equilibrium $\sigma^n$ of the game $\Gamma(B_n,n)$. Let $\sigma^n_i(L)=\lambda_n$ (i.e., $\lambda_n$ denotes the probability that a low type  player assigns to the action $1$).
By taking a sub sequence if neccessary,  we can assume that  the limit of $\{\lambda_n\}_{n=1}^{\infty}$ exists and equals $\lambda$, i.e. $$\lim_{n\rightarrow}\lambda_n=\lambda.$$
For any $\lambda_n$, let $\lambda_\omega(\lambda_n)$ denote the  probability that a player chooses  action $1$ conditional on  state  $\omega$, that is
\begin{equation}\label{eq:def_lambda_state}
\lambda_\omega(\lambda_n)=\begin{cases}p+(1-p)\lambda_n&\mbox{ if }\omega=H\\ (1-p)+p\lambda_n&\mbox{ if }\omega=L
\end{cases}.
\end{equation}
We consider two cases.

\textbf{Case 1. $q \leq 1-p:~~$}
First we show that $\lambda=0$. Assume by contradiction that $\lambda>0$. By equation \eqref{eq:def_lambda_state}
$$\lambda_H(\lambda)>\lambda_L(\lambda)\Rightarrow x_n>y_n.$$
Where $x_n$ is the probability that the good is supplied conditional on on state $\omega=H$ and $a_i=1$;  and $y_n$  is the probability that the good is supplied conditional on on state $\omega=L$ and $a_i=1$,  as defined in Definition \ref{def:xnyn}. 
By Theorem \ref{thm:unique_eq} high type players surely commits and thus for large values of $n$ we have $$y_n>1-p>q.$$
By the law of large numbers
$$
\lim_{n\rightarrow\infty} y_n=\lim_{n\rightarrow\infty} Pr_{\sigma^n}(\sum_{i=1}^{n-1} a_i/(B_n-1)>=1|\omega=L)=
\lim_{n\rightarrow\infty} Pr_{\sigma^n}(\sum_{i=1}^{n-1} a_i/n>=q|\omega=L)=1.
$$
Therefore since $x_n\ge y_n$ we have that for all large enough values of $n$
\begin{equation}\label{eq:neg_payoff}
(1-p)x_n-p y_n<0.
\end{equation}

As previously discussed,  equation \eqref{eq:neg_payoff} represents the expected low type payoff from playing the pure action$a=1$.  Hence, for large enough $n$ playing $a=1$ with positive probability yields a negative expected payoff to the low type.
This stands in contradiction to $\lambda=\lim_{n} \lambda_n>0$ hence we must have that $\lambda=0$.


\textbf{Case 2. $q \ge 1-p:~~$}
As before $\lambda=\lim_{n\rightarrow\infty} \lambda_n$. Assume by way of contradiction that $\lambda\neq q-(1-p)/p$.  We consider two sub-cases:


\textbf{Case 2.1.}
\begin{equation}\label{eq:limiting}
\lambda>\frac{q-(1-p)}{p}>0.
\end{equation}
In this case,   $\lambda_n>0$ for all large enough values of $n$. Therefore,
\begin{equation}\label{eq:ttt11}
\lim_{n\rightarrow\infty}\lambda_L(\lambda_n)=\lambda_L(\lambda)=(1-p)+p\lambda>q
\end{equation}
and
\begin{equation}\label{eq:ttt12}
\lambda_H(\lambda_n)>\lambda_L(\lambda_n).
\end{equation}
Equation \eqref{eq:ttt12} entails $x_n> y_n.$ Together with equation \eqref{eq:ttt11} it yields

\begin{equation}\label{eq:lim}
\lim_{n\rightarrow\infty} y_n=\lim_{n\rightarrow\infty} Pr_{\sigma^n}(\sum_{i=1}^{n-1} a_i/(B_n-1)>=1|\omega=L)=
\lim_{n\rightarrow\infty} Pr_{\sigma^n}(\sum_{i=1}^{n-1} a_i/n>=q|\omega=L)=1.
\end{equation}
By equation \eqref{eq:lim}  the expected payoff of a low type player from playing the action $a=1$ eventually becomes strictly negative, i.e.  $$(1-p)x_n-p y_n<0.$$
Therefore, as in Case 1 we get that for large values of $n$ the low type players can profitably deviate and play $a=0$ with probability one.


\textbf{Case 2.2.}  $\lambda<\frac{q-(1-p)}{p},$  The fact that $\lambda_n<\frac{q-(1-p)}{p}$ imply that  $\lambda_L(\lambda_n)$ eventually becomes smaller then $q.$   Hence by the by the Law of Large Numbers,
$$\lim_{n\rightarrow\infty} y_n=\lim_{n\rightarrow\infty}Pr_{\sigma^n}(\sum_{i=1}^{n-1} a_i/(B_n-1)>=1|\omega=L)=
\lim_{n\rightarrow\infty}Pr_{\sigma^n}(\sum_{i=1}^{n-1} a_i/n>=q|\omega=L)=0. $$

Since $\lim_{n\rightarrow\infty}\lambda_H(\lambda_n)= \lambda_H(\lambda)>\lambda_L(\lambda)$ we have that $x_n>y_n$ for all sufficiently large $n$.

Therefore we must also have that
$\lim_{n\rightarrow\infty} x_n=1$
$$\lim_{n\rightarrow\infty}\frac{Pr_{\sigma^n}(\sum_{i=1}^n a_i\geq B|a_i=1,\omega=L)}{Pr_{\sigma^n}(\sum_{i=1}^n a_i\geq B|a_i=1,\omega=H)}=\lim_{n\rightarrow\infty}\frac{y_n}{x_n}=0.$$
Since  $p>1/2$ we get that for all sufficiently large $n$
$$(1-p)x_n-py_n>0.$$

Once more we have reached a contradiction as there exist a profitable deviation to action $a=1$ for a low type player in the corresponding $\Gamma(B_n,n)$ for large $n$. 
The above two sub-cases show that $\lim_{n\rightarrow\infty}\lambda_n= [q-(1-p)]/p$
whenever $q\ge 1-p.$
\subsection{Proof of Lemma \ref{lem:diff_is_zero}}
Let $q\in[0,1]$ and $\{B_n\}_{n=1}^\infty$ be a sequence of thresholds with $\lim_{n\rightarrow\infty}\frac{B_n}{n}=q$. For every game in the sequence  $\{\Gamma(B_n,n)\}_{n=1}^\infty$
corresponds  unique non-trivial Bayesian equilibrium  $\sigma^n$. Let  $\sigma^n_i(L)=\lambda_n$ denote the probability that a low type  player assigns to the action $1$ in $\sigma^n$.  By taking a sub sequence if necessary, We can assume that  the following limit exists 
$$\lim_{n\rightarrow}\lambda_n=\lambda.$$

First assume that $\lambda<1$. Let $Pr_{\sigma^n}$ be the probability distribution  induced by $\sigma^n$ and recall that
\begin{equation*}
\begin{split}
&x_n=Pr_{\sigma^n}(\sum_{j=1}^{n-1} a_j\ge B_n-1|\omega=H)=Pr_{\sigma^n}(\sum_{j=1}^{n-1} a_j+1\ge B_n|\omega=H)=\\
&Pr_{\sigma^n}(\frac{\sum_{j=1}^{n-1} a_j}{n}+\frac{1}{n}\ge \frac{B_n}{n}|\omega=H).
\end{split}
\end{equation*}
Therefore
\begin{equation*}
\begin{split}
&\lim_{n\rightarrow\infty} x_n-Pr_{\sigma^n}(\sum_{i=1}^n a_i\geq B_n|\omega=H)=\\
&\lim_{n\rightarrow\infty}Pr_{\sigma^n}(\frac{\sum_{j=1}^{n-1} a_j}{n}+\frac{1}{n}\geq q|\omega=H,a_1=1)-Pr_{\sigma^n}(\frac{\sum_{i=1}^n a_i}{n}\geq q|\omega=H)=0.
\end{split}
\end{equation*}
A similar consideration holds for $y_n$ for $\omega=L$.

In the complementary case where $\lambda=1$ we have that
\begin{align*}
&\lim_{n\rightarrow\infty} Pr_{\sigma^n}(\sum_{i=1}^n a_i\geq B_n|\omega=H)=\\
&\lim_{n\rightarrow\infty} \lambda_nPr_{\sigma^n}(\sum_{i=1}^n a_i\geq B_n|\omega=H,a_1=1)+(1-\lambda_n)Pr_{\sigma^n}(\sum_{i=1}^n a_i\geq B_n|\omega=H,a_i=0)=x_n.
\end{align*}
The last equality holds by the definition of $x_n$ and as $\lim_{n\rightarrow\infty}\lambda_n=1$.
Similar consideration holds with respect to $y_n$ for $\omega=L.$

\subsection{Proof of Theorem \ref{thm:asym_correctness}}
\begin{proof}
	First we show that $\lim_{n\rightarrow\infty} \max_{B\in\{1\dots n+1\}} \theta(B,n) \ge \frac{3p-1}{2p}.$
	Consider the sequence $$\{B_n\}_{n=1}^{\infty}\mbox{  such that }B_n=\frac{n}{2}$$ and the corresponding sequence of games $\{\Gamma(\frac{n}{2},n)\}_{n=1}^\infty$ with $\sigma^n$ as the unique symmetric non-trivial equilibrium.
	
	As $\lim_{n\rightarrow\infty} \frac{B_n}{n}=\frac{1}{2}$ by Lemma \ref{lem:asympt_probs} and since $p>\frac{1}{2}$ we get  $$\lim_{n\rightarrow\infty}\sigma^n(L)=\lim_{n\rightarrow\infty} \lambda_n = \frac{\frac{1}{2}-(1-p)}{p}= \frac{2p-1}{2p}.$$
	Therefore $\lambda_H(\lambda_n)$ converges to $\frac{3p-1}{2p}>\frac{1}{2}$ which by the law of large numbers implies $\lim_{n\rightarrow\infty} x_n =1$. Since   $\lambda_n$ eventually becomes strictly positive,  the indifference condition of the low-type player implies that
	\begin{equation}\label{eq:low_ind1}
	(1-p)x_n-py_n=0.
	\end{equation}
	$x_n$ converges to $1$ and thus by equation \eqref{eq:low_ind1} we get $$\lim_{n\rightarrow\infty} y_n = \frac{1-p}{p}.$$
	By Lemma \ref{lem:diff_is_zero} we show that,
	$$\lim_{n\rightarrow\infty}\Theta(\frac{n}{2},n)=\frac{1}{2}+\frac{1}{2}(1-\frac{1-p}{p})=\frac{3p-1}{2p}.$$
	And thus
	$\lim_{n\rightarrow\infty} \max_{B\in\{1\dots n+1\}} \theta(B,n) \ge \frac{3p-1}{2p}.$

	Next we show the opposite inequality holds as well , i.e. $$\lim_{n\rightarrow\infty} \max_{B\in\{1\dots n+1\}} \theta(B,n) \le \frac{3p-1}{2p}.$$
	Consider a sequence $\{B_n\}_{n=1}^\infty$ of thresholds such that the sequences
	$\{\theta(B_n,n)\}_{n=1}^\infty, \{x_n\}_{n=1}^\infty$ and $\{y_n\}_{n=1}^\infty$ converge to  $\theta^*, x^*$ and $y^*$ respectively.
	
	By definition of the correctness index and Lemma \ref{lem:diff_is_zero},
	\begin{equation}\label{eq:lim_xnyn}
	\lim_{n\rightarrow\infty} \theta(B_n,n) = \frac{1}{2}x^* + \frac{1}{2}(1-y^*).
	\end{equation}

	We show next that $\frac{3p-1}{2p}$ is the best possible asymptotic correctness value.
	Let $\{(B^*_n,n)\}_{n=1}^{\infty}$ be a sequence of thresholds that achieves the best asymptotic correctness probability. By taking subsequence we can assume that $\lim_{n\rightarrow\infty} \frac{B^*_n}{n}=q^*$ for some $q^*\in[0,1]$.
	
	First we note that $q^>(1-p)$. Assume to the contrary that $q\leq 1-p$. By Theorem \ref{thm:unique_eq} for all $B^*_n,n$ and in the corresponding $\Gamma(B^*_n,n)$, high type players surely choose action $1$ in non-trivial equilibrium and thus by the Law of large numbers,
	\begin{equation}
	\lim_{n\rightarrow\infty} x_n = \lim_{n\rightarrow\infty} y_n= 1.
	\end{equation}
	By Lemma \ref{lem:diff_is_zero} and by the definition of the asymptotic correctness index we have
	\begin{equation}
	\lim_{n\rightarrow\infty} \theta(B^*_n,n)=\frac{1}{2}.
	\end{equation}
	And thus we have reached a contradiction as the sequence $\{\frac{n}{2}\}_{n=1}^{\infty}$ yields an asymptotic correctness value of $\frac{3p-1}{2p}>\frac{1}{2}$ which is greater then the asymptotic correctness value  of the sequence $\{B^*_n\}_{n=1}^{\infty}$.
	And thus it must be the case that  $q^*\in(1-p,1]$.
	
	As $q^*\in(1-p,1]$, by Lemma \ref{lem:asympt_probs} we get that $\sigma^n_i(L)=\lambda_n$ is positive from some $n$ and on, and thus by the indifferent condition of a low type player we get
	\begin{equation}\label{eq:L_indif}
	(1-p)x_n-py_n=0.
	\end{equation}

	By taking another subsequence if necessary, we can assume that $\{x_n\}$ and $\{y_n\}$ converges to the limits $x^*$ and $y^*$ respectively.
	By Lemma \ref{lem:diff_is_zero}, $$\lim_{n\rightarrow\infty} \max_{B\in\{1\dots n\}} \theta(B_n,n) =px^*+(1-p)(1-y^*).$$
	And  by equation \eqref{eq:L_indif} $$(1-p)x^*-py^*=0.$$
	Therefore the asymptotic correctness is bounded by the linear programming below.
	\begin{equation*}
	\begin{aligned}
	& {\text{max}}
	& & \frac{1}{2}x+\frac{1}{2}(1-y) \\
	& \text{s.t.} & & 1\geq x,y\geq 0 \\
	& & &  (1-p)x-py=0. \\
	\end{aligned}
	\end{equation*}
	Which is simply
	$$\max_{x\in[0,1]}\frac{1}{2}x+\frac{1}{2}(1-\frac{1-p}{p}x)=\max_{x\in[0,1]}\frac{2px+p-x}{2p}=\frac{3p-1}{2p}.$$
\end{proof}

\subsection{Proof of Theorem \ref{thm:revenue}}

First we show that the firm can achieve a Market penetration index of $\frac{1}{2p}$.
Let $q\in[(1-p),1)$. Consider a sequence of thresholds $\{B_n\}_{n=1}^{\infty}$ such that $B_n=q n$ and the corrsponding sequence of games $\{\Gamma(q n,n)\}_{n=1}^{\infty}$ with $\{\sigma^n\}$ as the unique symmetric non-trivial equilibrium.
Let $\lambda_n=\sigma^n_i(L)$ and let $x_n,y_n$ be as in Definition \ref{def:xnyn}.

By Lemma \ref{lem:asympt_probs}, for $q>1-p$
$$\lim_{n\rightarrow\infty}\lambda_n=\lambda=\frac{q-(1-p)}{p}>0$$
and thus for sufficiently large $n$, the low type player is indifferent between playing either $a=1$ or $a=0$ . Therefore
\begin{equation}\label{eq:ind_low_type}
(1-p)x_n-py_n=0
\end{equation}
and
\begin{equation}\label{eq:lim_lambda_h}
\lim_{n\rightarrow\infty}\lambda_H(\lambda_n)\frac{q(1-p)+2p-1}{p}>q>0.
\end{equation}
By the Law of large numbers, equation \eqref{eq:lim_lambda_h} entails that $\lim_{n\rightarrow\infty}x_n=1.$ This, together with equation \eqref{eq:ind_low_type}, yields that  the sequence $\{y_n\}_{n=1}^\infty$ converges to $\frac{1-p}{p}.$ I.e.,
\begin{equation}\label{eq:xnyn_limits}
\begin{split}
&x_n\underset{n\rightarrow \infty}{\longrightarrow} 1\\
&y_n\underset{n\rightarrow \infty}{\longrightarrow} \frac{1-p}{p}.\\
\end{split}
\end{equation}
By Lemma \ref{lem:diff_is_zero}, equation \eqref{eq:xnyn_limits} and the definition of the Market Penetration index we get that for every $q\in[(1-p),1)$
$$\lim_{n\rightarrow\infty}R(qn,n)=\frac{1}{2}+\frac{1}{2}(\frac{1-p}{p})=\frac{1}{2p}.$$
Therefore for large values of $n$,  $R(qn,n)\geq\frac{q}{2p}$ and thus the firm can guarantee a market penetration index of $\frac{1}{2p}$ by taking $q$ to one.

To complete the proof we show that $\frac{1}{2p}$ is indeed the firm's optimal asymptotic  Market penetration index. Consider a sequence of  thresholds $\{B^*_n\}_{n=1}^{\infty}$, the corresponding sequence of games $\Gamma(B_n,n)$ with  $\{\sigma^n\}_{n=1}^{\infty}$ as the unique symmetric non-trivial  equilibrium for which the maximum in Equation \eqref{eq:revenue} is attained.
Let $\lambda_n=\sigma^n_i(L)$ and $x_n,y_n$ be as defined in Definition \ref{def:xnyn}. By taking a subsequence if necessary we can assume that the following limit exists for some $q^*\in[0,1]$,
$$
\lim_{n\rightarrow\infty} \frac{B^*_n}{n}=q^*
$$.

Next we prove that $q^*\in[1-p,1)$. Assume by contradiction that $q^*<1-p$ and thus by Lemma \ref{lem:asympt_probs}  $\lambda_n=0$ for all sufficiently large values of $n$.
In this case $$\lim_{n\rightarrow\infty}\lambda_L(\lambda_n)=1-p>q^*\mbox{ and } \lim_{n\rightarrow\infty}\lambda_H(\lambda_n)=p>q^*,$$
and thus both
$\{x_n\}_{n=1}^{\infty}$ and $\{y_n\}_{n=1}^{\infty}$ converge to $1.$ Therefore, by the definition of the Market penetration
$$\lim_{n\rightarrow\infty}R(B^*_n,n)=\frac{1}{2}(1-p)+\frac{1}{2}p=\frac{1}{2}<\frac{1}{2p}.$$
We have reached a contradiction as we show above that the value $\frac{1}{2p}$ is achievable. 

As $q^*\in[1-p,1)$ and  by Lemma \ref{lem:asympt_probs} for sufficiently large values of $n$,  $\lambda_n>0$. By taking sub-sequence we can assume that $\lambda_n>0$ for all $n$.
In this case, by the indifference condition of the low type player at the corresponding $\Gamma(B^*_n,n)$, and for every $n$  we have that
$$(1-p)x_n-py_n=0.$$

By taking a subsequence if necessary, we can assume that $\{x_n\}_{n=1}^{\infty}$ and $\{y_n\}_{n=1}^{\infty}$ converge to some $x^*$ and $y^*$ respectively and thus
$$(1-p)x^*-py^*=0.$$
Note that the maximal  Market penetration for the firm conditional on producing the product is bounded by one. This upper bound on the Market penetration index, together with Lemma \ref{lem:diff_is_zero} entails
$$\lim_{n\rightarrow\infty} \max_{B\in\{1\dots n\}} R(B,n)\leq \frac{1}{2}x^*+\frac{1}{2}y^*.$$
Therefore, the optimal  market penetration for the firm is bounded by the solution of the following linear programming:
\begin{equation}
\begin{aligned}
& {\text{max}}
& & \frac{1}{2}x+\frac{1}{2}y \\
& \text{s.t.} & & 1\geq x,y\geq 0 \\
& & &  (1-p)x-py=0. \\
\end{aligned}
\end{equation}
Which is simply
$$\max_{x\in[0,1]}\frac{1}{2}x+\frac{1}{2}(\frac{1-p}{p}x)=\max_{x\in[0,1]}\frac{x}{2p}=\frac{1}{2p}.$$
\qed

\end{document}